\documentclass{article}
\usepackage{authblk}
\usepackage{amsmath,amssymb,amsthm}
\usepackage[utf8]{inputenc}
\usepackage[T1]{fontenc}
\usepackage{geometry}
\usepackage{lmodern}
\usepackage[numbers]{natbib} 
\usepackage[breaklinks,colorlinks]{hyperref}
\usepackage{mathtools}
\usepackage{microtype}
\usepackage{graphicx}
\usepackage{tikz}
\usetikzlibrary{positioning,backgrounds,patterns,calc}
\usepackage{booktabs}
\usepackage{comment}
\usepackage{color}
\usepackage{xspace}
\usepackage[textsize=footnotesize,textwidth=1.7cm]{todonotes}
\usepackage{lineno}

\colorlet{lblue}{white!25!blue}
\colorlet{grey}{white!50!black}
\colorlet{dblue}{black!50!lblue!70!grey}

\newtheorem{theorem}{Theorem}
\newtheorem{lemma}{Lemma}
\newtheorem{corollary}{Corollary}
\newtheorem{definition}{Definition}
\newtheorem{rrule}{Rule}
\newtheorem{brule}{Branching Rule}

\newcommand{\Oh}{\mathcal{O}}
\newcommand{\LGM}{\textsc{List-Colored Graph Motif}\xspace}
\newcommand{\sLGM}{\textsc{LGM}\xspace}
\newcommand{\GM}{\textsc{Graph Motif}\xspace}
\newcommand{\sGM}{\textsc{GM}\xspace}
\newcommand{\CGM}{\textsc{Colorful Graph Motif}\xspace}
\newcommand{\sCGM}{\textsc{CGM}\xspace}

\newcommand{\num}{\occ}
\newcommand{\MCC}{\textsc{Multicolored Clique}\xspace}
\newcommand{\halfskip}{\ensuremath{\mskip0.5\thinmuskip{}}}
\renewcommand{\L}{\mathcal{L}\xspace}

\DeclareMathOperator{\occ}{occ}
\newcommand{\mult}{M}

\title{Graph Motif Problems Parameterized by Dual\footnote{A
    preliminary version of this work appeared in \emph{Proceedings  of
      the 27th Annual Symposium on Combinatorial Pattern Matching
      ({CPM}~'16)}, volume 54 of LIPIcs, pages 7:1--7:12. This version
    contains all missing proofs, an improved running time for
    Theorem~\ref{thm:gm-tree} and a new tractability result
    (Theorem~\ref{thm:lgm-tree-H-bounded-mult}). CK was partially
    supported by the DFG, project ``Multivariate 
    algorithmics for graph and string problems in bioinformatics'' (KO 3669/4-1).}}

\author[1]{Guillaume Fertin}
\affil[1]{Laboratoire des Sciences du Numérique de Nantes,
  UMR CNRS 6004, Universit\'e de Nantes, 2 rue de la Houssini\`ere, 44322
  Nantes Cedex 3, France\\ \href{mailto:guillaume.fertin@univ-nantes.fr}{guillaume.fertin@univ-nantes.fr}
}
\author[2]{Christian Komusiewicz}
\affil[2]{Fachbereich Mathematik und Informatik, Philipps-Universität Marburg,
  Germany\\\href{mailto:komusiewicz@informatik.uni-marburg.de}{komusiewicz@informatik.uni-marburg.de}} 

\begin{document}

\maketitle
\begin{abstract}
  Let $G=(V,E)$ be a vertex-colored graph, where $C$ is the set of
  colors used to color $V$. The \GM (or \sGM)
  problem takes as input $G$, a multiset $M$ of colors built from $C$,
  and asks whether there is a subset $S\subseteq V$ such that
  (i)~$G[S]$ is connected and (ii)~the multiset of colors obtained
  from $S$ equals $M$. The \CGM (or \sCGM) problem is the
  special case of \sGM in which $M$ is a set, and the
  \LGM (or \sLGM) problem is the extension of \sGM in which each vertex
  $v$ of $V$ may choose its color from a list $\L(v)\subseteq C$ of colors.

  We study the three problems \sGM, \sCGM, and \sLGM,
  parameterized by the dual parameter~$\ell:=|V|-|M|$. For general graphs,
  we show that, assuming the strong exponential time hypothesis, \sCGM
  has no $(2-\epsilon)^\ell\cdot |V|^{\Oh(1)}$-time algorithm, which implies
  that a previous algorithm, running in $\Oh(2^\ell\cdot |E|)$~time is optimal~[Betzler et al., IEEE/ACM TCBB 2011]. We also prove
  that \sLGM is W[1]-hard with respect to~$\ell$ even if we restrict ourselves to lists of at
  most two colors. If we constrain the input graph to be a tree, then we
  show that \sGM can be solved in $\Oh(3^\ell\cdot |V|)$ 
  time but admits no polynomial-size problem
  kernel, while \sCGM can be solved in $\Oh(\sqrt{2}^{\halfskip\ell} + |V|)$
  time and admits a polynomial-size problem kernel.
\end{abstract}

\section{Introduction}

The \textsc{Subgraph Isomorphism} problem is the following pattern matching problem in graphs: given a (typically large)
host graph $G$ and a (small) query graph $H$, return one (or all) occurrence(s) of $H$ in $G$,
where the term occurrence denotes here a subset $S$ of $V(G)$ such that
$G[S]$, the subgraph of $G$ induced by $S$, is isomorphic to $H$. This
type of graph mining problem has different applications, notably in
biology~\cite{SI06}. 
\textsc{Subgraph Isomorphism} is a {\em structural} graph pattern matching problem, where one
looks for similar graph structures between $H$ and $G$. In some biological contexts, however,
additional information is provided to the vertices of the graphs, for example their biological
function. This can be modeled by labeling each vertex of the graph, for example by giving it
one or several colors, each corresponding to an identified function. In the presence of such
functional annotation, the structure of a given induced subgraph may be of less importance than
the functions it corresponds to. Thus, a new set of {\em functional} graph pattern matching
problems has emerged, starting with the \GM problem~\cite{LFS06}, which was introduced in the
context of the analysis of metabolic networks. In~\textsc{Graph Motif}, the query is a
multiset $M$ of colors that represents the functions of interest, and we search for an
occurrence of~$M$ in the host graph, where the previous demand of being isomorphic to the query is
replaced by a connectivity demand.

\begin{quote}
  \textsc{Graph Motif} (\textsc{GM})\\
  \textbf{Input:} A multiset $M$ built on a set $C$ of colors, an
  undirected graph~$G=(V,E)$, and a coloring $\chi: V\to C$.\\
  \textbf{Question:} Is there a set~$S\subseteq V$ such
  that~$G[S]$ is connected and there is a one-to-one mapping~$f:S\to M$ such that~$f(v)=\chi(v)$ for all~$v\in S$?
\end{quote}

Many variants of \sGM have been introduced and studied. In
particular, \LGM (or \sLGM) is a generalization of \sGM that is used
to identify, in a given protein interaction network, protein complexes that are
similar to a given protein complex from a different
species~\cite{BHKSS10}. In \sLGM, the graph~$G$ is associated with a list-coloring~$\L: V\to 2^C$, that is, each vertex~$v$ is associated with a set~$\L(v)$ of colors, and the question is whether there is a set~$S\subseteq V$ such that (i)~$G[S]$ is connected and
(ii)~the one-to-one mapping~$f$ from~$S$ to~$M$ we look for satisfies
$\forall v\in S: f(v) \in \L(v)$. The special case of \textsc{GM} in
which $M$ is a set is called \textsc{Colorful Graph Motif} (or
\textsc{CGM}). Many optimization problems related to \sGM have
received interest, including some that are related to tandem mass
spectrometry and where the input graph is directed and edge-weighted~\cite{RRNB13}. All
these problem variants have given rise to a very abundant
literature. \sCGM, \sGM, and \sLGM are NP-hard even in very restricted
cases~\cite{LFS06,FFHV11,BS15}. Consequently, many of the above-mentioned studies
have focused on (dis)proving fixed-parameter tractability of the
problems (see e.g.~\cite{sikora-survey12} for an informal survey
on the topic). In such cases, very often the parameter $k:=|M|=|S|$ is
considered.

In this paper, we study the parameterized complexity of \sGM, \sCGM,
and \sLGM, but we differ from the usual viewpoint by focusing on the
{\em dual} parameter $\ell:=|V|-|S|$, that is,
$\ell$ is the number of vertices to be deleted from $G$ to obtain a
solution. Although the choice of $\ell$ may be disputable because {\em a priori} it
may be too large to expect a good behavior in practice,
there are several arguments for choosing such a parameter: First,
after some initial data reduction, the input may be divided into
smaller connected components, where~$\ell$ is not much larger
than~$k$. Second, the algorithms for parameter~$k$ rely on algebraic
techniques or dynamic programming, and in both cases, the worst-case
running time is equivalent to the actual running time. In contrast,
for example for \sCGM, the algorithm for parameter~$\ell$ is a search
tree algorithm~\cite{BBFKN11}, and search tree algorithms can be
substantially accelerated via pruning rules. Finally, there are
subgraph mining problems where the dual parameter~$\ell$ is usually
bigger than the parameter~$k$ but leads to the current-best algorithm
(in terms of performance on real-world instances), see e.g.~\cite{HKN15}. Hence,
parameterization by~$\ell$ may be useful even if~$\ell$ is bigger
than~$k$, and thus deserves to be studied.

\subparagraph*{Related work and our contribution.}  \sGM is NP-hard,
even when $M$ is composed of two colors~\cite{FFHV11}. Concerning the
parameterized complexity for parameter $k:=|M|$, the current-best
randomized algorithm has a running time of $2^k\cdot
n^{\Oh(1)}$~\cite{BKK13,PZ14} where $n:=|V|$, and there is evidence
that this cannot be improved to a running time of $(2-\epsilon)^k\cdot
n^{\Oh(1)}$~\cite{BKK13}. The current-best running time for a
deterministic algorithm is~$5.22^k\cdot n^{\Oh(1)}$~\cite{PSZ16}.  
\sGM on trees can be solved in~$n^{\Oh(c)}$ time where~$c$ is the number of
colors in $M$~\cite{FFHV11}, but is W[1]-hard with respect
to~$c$~\cite{FFHV11}. Other parameters, essentially related to the
structure of the input graph $G$, have been studied by
Ganian~\cite{Gan11}, Bonnet and Sikora~\cite{BS15}, and Das et
al.~\cite{DEMR17}. For example, \GM is fixed-parameter tractable when
parameterized by the size of a vertex cover of the input
graph~\cite{Gan11,BS15}.  
Finally, concerning parameter~$\ell$,~\sGM has been shown to be
W[1]-hard, even when $M$ is composed of two colors~\cite{BBFKN11}.  

Since \sCGM is a special case of \sGM, any above-mentioned
positive result for \sGM also holds for \sCGM.  In addition, \sCGM is
NP-hard even for trees of maximum degree~3~\cite{FFHV11}, and does
not admit a polynomial-size problem kernel with respect to~$k$ even
if~$G$ has diameter two or if~$G$ is a comb graph (a special type of
tree with maximum degree 3)~\cite{ABH+10}. Finally, 
\sCGM can be solved in $\Oh(2^\ell\cdot m)$
time~\cite{BBFKN11}, where $m:=|E|$.  The \sLGM problem is an extension of~\sGM{} and
thus any negative result for \sGM propagates to
\sLGM. Moreover,~\sLGM{} is  fixed-parameter tractable with
respect to~$k$, and the current-best algorithm runs in $2^k\cdot n^{\Oh(1)}$
time~\cite{PZ14}. Concerning parameter $\ell$, \sLGM has been shown to
be W[1]-hard even when $M$ is a set~\cite{BBFKN11}.

As mentioned above, we study \sGM, \sLGM and \sCGM with respect to the
dual parameter~$\ell:= n-k$. Since many results in general graphs turn
out to be negative, we also study the special case where
the input graph $G$ is a tree. Our results are summarized in
Table~\ref{tab:results}. In a nutshell, we strengthen previous
hardness results for the general case and show that the $\Oh(2^\ell\cdot
m)$-time algorithm for \sCGM{} is essentially optimal. Then, we show
that for~\sGM{} on trees and for some special cases of~\sLGM{} on
trees, a fixed-parameter algorithm can be achieved. Finally, we show
that for~\sCGM{} on trees, a polynomial-size
problem kernel and better
running times than for general graphs can be achieved.

\begin{table}[t]
  \caption{Overview of new and previous results with respect to the dual
    parameter~$\ell:=n-k$, where~$n:=|V|$,~$k:=|M|$,~$m:=|E|$ and $\Delta:=\max_{v\in V} |\L(v)|$ denotes the
    maximum list size in~$G$. The lower bound result for~\sCGM assumes the
    strong exponential time hypothesis (SETH)~\cite{IPZ01}. }  \centering 
\begin{tabular}[t]{r c c}\toprule  & General graphs & Trees \\ \midrule
  LGM~~ & W[1]-hard~\cite{BBFKN11} & ? \\
  LGM,~$\Delta=2$ & W[1]-hard (Cor.~\ref{cor:lb-small-list}) & ? \\
  \midrule
  GM~~  & W[1]-hard~\cite{BBFKN11} & $\Oh(3^{\ell}\cdot
  n)$ (Thm.~\ref{thm:gm-tree})\\ & & no poly.~kernel~(Thm.~\ref{thm:gm-tree-no-kernel}) \\ \midrule
  CGM~~  & $\Oh(2^{\ell}\cdot m)$~\cite{BBFKN11}, & $\Oh(\sqrt{2}^{\halfskip\ell} + n)$~(Thm.~\ref{thm:cgm-tree}), \\
  & no~$(2-\epsilon)^{\ell}\cdot n^{\Oh(1)}$~(Thm.~\ref{thm:lb-dual}) &  \\
  & no poly.~kernel~(Thm.~\ref{thm:lb-kernel}) & $(2\ell+1)$-vertex kernel~(Thm.~\ref{thm:cgm-tree-kernel}) \\ \bottomrule
\end{tabular}
\label{tab:results}
\end{table}

\subparagraph*{Preliminaries.}  For an integer~$n$, we use~$[n]:=\{1,\ldots,n\}$ to denote the
set of the integers from~$1$ through~$n$.  Throughout the paper, the input graph for our three
problems is $G=(V,E)$, and we let~$n:=|V|$ (resp. $m:=|E|$) denote its number of vertices
(resp. edges). For a vertex set~$S\subseteq V$, we let~$G[S]:=(S,\{\{u,v\}\mid u,v\in S\})$ denote the subgraph induced by~$S$. 
The set~$S$ of vertices sought for in the three problems is called an
\emph{occurrence} of~$M$.  If~$G$ is vertex-colored, we call a vertex set~$S$ \emph{colorful}
if~$|S|=|M|$ and all vertices in~$S$ have pairwise different colors. A vertex~$v$ is called
\emph{unique} if it is assigned a color $c$ that is assigned to no other vertex in $V$. For a
multiset~$M$ and an element~$c$ of~$M$, we use~$M(c)$ to denote the multiplicity of~$c$ in~$M$.

To analyze the structure of the coloring constraints for instances of~\sLGM, we consider the following auxiliary graph.
\begin{definition}
  Let~$(M,G,\L)$ be an instance of~\sLGM. The vertex-color graph~$H$
  of~$(M,G,\L)$ is the bipartite graph with vertex set~$V\cup C$ and
  edge set~$\{\{v,c\}\mid v\in V, c\in \L(v) \}$.
\end{definition}
Observe that \sGM~instances are~\sLGM~instances where in the vertex-color graph~$H$ each vertex
from~$V$ has degree one. In other words, $H$ is a disjoint union of stars whose non-leaf is a
vertex from~$C$.  Moreover, an~\sLGM~instance where~$H$ is a disjoint union of bicliques can be
easily replaced by an equivalent~\sGM~instance: For each biclique~$K$ in~$H$, replace the color
set~$K\cap C$ by one color with multiplicity~$\sum_{c\in C} M(c)$ in~$M$ and assign this color
to all vertices in~$K\cap V$.

We briefly recall the relevant notions of parameterized algorithmics~\cite{CFK+15,DF13}.
Parameterized algorithmics aims at analyzing the impact of structural input properties on the
difficulty of computational problems. Formally, a parameterized problem~$L$ is a subset
of~$\Sigma^*\times \mathbb{N}$ where the first component is the input instance and the second component is the \emph{parameter}. A parameterized problem~$L$ is fixed-parameter tractable if every input instance~$(I,k)$ can be solved in~$f(k)\cdot |I|^{\Oh(1)}$ time where~$f$ is a computable function depending only on~$k$.
A \emph{reduction to a problem kernel}, or \emph{kernelization},
is an algorithm that takes as input an instance~$(I,k)$ of a parameterized problem and produces
in polynomial time an instance~$(I',k')$ such that
\begin{itemize}
\item $(I,k)$ is a yes-instance if and only if~$(I',k')$ is a yes-instance and
\item $|I'|\le g(k)$ where~$g$ is a computable function depending only on $k$.
\end{itemize}
The instance~$(I',k')$ is called \emph{problem
  kernel} and~$g$ is called the \emph{size of the problem kernel}. If~$g$ is a polynomial
function, then the problem admits a \emph{polynomial-size problem kernelization}.  The class
W[1] is a basic class of presumed fixed-parameter intractability~\cite{CFK+15,DF13}, that is, if a
problem is W[1]-hard for parameter~$k$, then we assume that it cannot be solved
in~$f(k)\cdot n^{\Oh(1)}$ time~\cite{CFK+15,DF13}. The strong exponential time hypothesis (SETH)
assumes that, for any~$\epsilon>0$, CNF-SAT cannot be solved in
time~$(2-\epsilon)^n\cdot |\Phi|^{\Oh(1)}$ where~$\Phi$ is the input formula and~$n$ is the
number of variables~\cite{IPZ01}.

This work is structured as follows. In Section~\ref{sec:lb-dual}, we present
 lower bounds for~\sLGM{} and~\sCGM{} on general
graphs. These negative results motivate our study of the case when~$G$ is a tree; our results for~\sGM{} on trees and~\sCGM{} on
trees will be presented in~Section~\ref{sec:gm} and~Section~\ref{sec:cgm},
respectively.  We conclude with an outlook of future work
 in~Section~\ref{sec:conclusion}.

\section{Parameterization by Dual in General Graphs: Tight Lower Bounds}
\label{sec:lb-dual}
\sCGM{} can be solved in~$\Oh(2^\ell\cdot m)$ time~\cite{BBFKN11}. We
show that this running time bound is essentially optimal.
\begin{theorem}
\label{thm:lb-dual}
\CGM cannot be solved in~$(2-\epsilon)^\ell \cdot n^{\Oh(1)}$ time
unless the strong exponential time hypothesis (SETH) fails.
\end{theorem}
\begin{proof}
  We present a polynomial-time reduction from CNF-SAT:
  \begin{quote}
    \textbf{Input:} A boolean formula~$\Phi$ in conjunctive normal
    form with clauses~${\cal C}_1,\ldots , {\cal C}_q$ over variable
    set~$X=\{x_1,\ldots , x_r\}$.\\
    \textbf{Question:} Is there an assignment~$\beta: X\to \{\texttt{true},\texttt{false}\}$ that
    satisfies~$\Phi$?
  \end{quote}
  The reduction works as follows. First, for each variable~$x_i\in X$
  introduce two \emph{variable vertices}~$v^t_i$ and~$v^f_i$ and color
  each of the two vertices with color~$\chi^X_i$. The idea is that
  (with the final occurrence) we must select exactly one vertex for
  this color. This selection will correspond to a truth assignment
  to~$X$. Now, introduce for each clause~${\cal C}_i$ a \emph{clause
    vertex}~$u_i$, color~$u_i$ with a unique color~$\chi^{\cal C}_i$
  and make~$u_i$ adjacent to vertex~$v^t_j$ if~$x_j$ occurs nonnegated
  in~${\cal C}_i$ and to vertex~$v^f_j$ if~$x_j$ occurs negated
  in~${\cal C}_i$. Finally, introduce one further vertex~$v^*$ with a
  unique color~$\chi^*$, make~$v^*$ adjacent to all variable vertices
  and let~$M$ be the set containing each of the introduced colors
  exactly once. Note that there are exactly~$|X|$ colors that appear
  twice in~$G$ and that all other colors appear exactly
  once. Hence,~$\ell=|X|$. We next show the correctness of the
  reduction. Let~$I=(M,G,\chi)$ denote the constructed instance of~\sCGM.
  
  First, assume that~$\Phi$ is satisfiable and let~$\beta$ be a
  satisfying assignment of~$X$. For the \sCGM instance~$I$ consider the
  vertex set~$S\subseteq V$ that contains all clause vertices,
  vertex~$v^*$, and for each variable~$x_i$ the vertex~$v^t_i$ if~$\beta$ 
  sets~$x_i$ to \texttt{true} and the vertex~$v^f_i$ otherwise. Clearly,~$|S|=|M|$ and
  no two vertices of~$S$ have the same color. To show that~$I$ is a
  yes-instance of \sCGM it remains to show that~$G[S]$ is connected.
  First, the subgraph induced by the variable vertices in $S$ plus~$v^*$ is a
  star and thus it is connected. Second, since~$\beta$ is a satisfying
  assignment each clause vertex in~$S$ has at least one neighbor
  in~$S$ (which is by construction a variable vertex). Hence,~$G[S]$
  is connected.

  Conversely, assume that~$I$ is a yes-instance of~\sCGM, and let~$S$ be a
  colorful vertex set with~$|S|=|M|$ such that~$G[S]$ is
  connected. Since~$S$ is colorful, the variable vertices in~$S$
  correspond to a truth assignment of~$X$. This assignment
  satisfies~$X$: Indeed, since $G[S]$ is connected, there is a path in~$G[S]$ between each clause vertex~$u_i$ and~$v^*$, 
  and thus there is a neighbor of~$u_i$ that is in~$S$. If this
  neighbor is~$v^t_j$ (resp.~$v^f_j$), then by construction,~$\beta$
  assigns \texttt{true} (resp. \texttt{false}) to~$x_j$ and thus~${\cal  C}_i$ is satisfied.

  Thus, the two instances are equivalent. Now observe that
  since~$\ell=|X|=r$ and $n=2r+q+1$, any~$(2-\epsilon)^\ell\cdot n^{\Oh(1)}$-time
  algorithm implies a~$(2-\epsilon)^r\cdot (r+q)^{\Oh(1)}$-time algorithm
  for CNF-SAT. This directly contradicts the SETH.
\end{proof}

The above reduction also makes the existence of a polynomial-size
problem kernel for parameter~$\ell$ unlikely. This is implied by
the following two facts. First, CNF-SAT parameterized by the number of
variables does not admit a polynomial-size problem kernel unless
$\textrm{NP}\subseteq \textrm{coNP/poly}$~\cite{DM10}. Second, the
reduction presented in the proof of Theorem~\ref{thm:lb-dual} is a
\emph{polynomial parameter transformation}~\cite{BTY11} from CNF-SAT
parameterized by the number of variables to \sCGM parameterized
by~$\ell$. More precisely, given an input CNF-SAT formula $\Phi$ on
variable set~$X$, the reduction produces an instance $I=(M,G,\chi)$ of~\sCGM{}
with~$\ell=|X|$. Now, any polynomial-size problem kernelization
applied to $I$ 
produces in polynomial time an equivalent~\sCGM{} instance $I'$ of
size~$\ell^{\Oh(1)}=|X|^{\Oh(1)}$. Since~CNF-SAT is NP-hard, we can now
transform this~\sCGM{} instance in polynomial time into an equivalent
CNF-SAT instance that has size~$\ell^{\Oh(1)}=|X|^{\Oh(1)}$.
Hence, a polynomial-size problem kernel for \sCGM{} parameterized by~$\ell$ implies a
polynomial-size problem kernel for CNF-SAT parameterized
by~$|X|$. This implies~$\textrm{NP}\subseteq
\textrm{coNP/poly}$~\cite{DM10} (which in turn implies a collapse of the
polynomial hierarchy).
\begin{theorem}
  \label{thm:lb-kernel}
  \CGM parameterized by~$\ell$ does not admit a polynomial-size
  problem kernel unless~$\text{\rm NP}\subseteq \text{\rm {coNP/poly}}$.
\end{theorem}
We have thus resolved the parameterized complexity of~\sCGM{}
parameterized by~$\ell$ on general graphs and now turn to the more
general \sLGM problem, which is~W[1]-hard with respect
to~$\ell$~\cite{BBFKN11}. Here, it would be desirable to obtain
fixed-parameter algorithms for parameter~$\ell$ at least for some
restricted inputs. In other words, we would like to further exploit
the structure of real-world instances to obtain tractability results.
A very natural approach here is to consider the size and structure of
the list-colorings~$\L(v)$ as additional parameter. Unfortunately, the
problem remains W[1]-hard even for the following very restricted case
of list-colorings. Recall, that the vertex-color graph is the
bipartite graph with vertex set~$V\cup C$ in which~$v\in V$ and~$c\in C$ are adjacent if and only if~$c\in \L(v)$.
\begin{theorem}\label{thm:lb-small-list}
  \LGM is W[1]-hard with respect to~$\ell$ even if the vertex-color
  graph is a disjoint union of paths.
\end{theorem}
\begin{proof}
  We reduce from the \textsc{Multicolored Independent Set}
  problem:
  \begin{quote}
    \textbf{Input:} An undirected graph~$H=(W,F)$ and a
    vertex-labeling~$\lambda: W\to \{1, \ldots, k\}$.\\
    \textbf{Question:} Is there a set~$S\subseteq W$ such
    that~$|S|=k$, the vertices in~$S$ have pairwise different labels and~$H[S]$ has no edges? 
  \end{quote}
  
  \textsc{Multicolored Independent Set} has been shown to be W[1]-hard
  when parameterized by~$k$~\cite{FHRV09}. We call the colors of the \textsc{Multicolored Independent Set} labels to avoid
  confusion with the colors of the \LGM~instance.  Assume without loss of generality that each
  label class in~$H$ contains the same number~$x$ of vertices (this can be achieved by padding smaller classes with additional vertices) and that there is an arbitrary
  but fixed ordering of the vertices of~$H$.

  The reduction works as follows. We first describe the input graph
  $G$ for \sLGM. We let $V=V_0\cup V_1\cup \{v^*\}$, where $V_0=W$ and
  $V_1=\{v_e | e\in F(H)\}$. Now construct the edge set~$E$ of~$G$ as
  follows. First, make vertex~$v^*$ adjacent to all vertices
  of~$V_0$. Then, for each edge~$\{u,w\}$ of~$H$ make
  vertex~$v_{\{u,w\}}$ adjacent to~$u$ and~$w$. This
  completes the construction of~$G$. 
  Now let us describe the coloring of the vertices. We start with the
  colors given to $V_0=W$. For each label~$i$ from $\lambda$ do the following:
  create~$x-1$ colors~$c^i_{1}, \ldots , c^i_{x-1}$. Now, with respect
  to the above-mentioned ordering, color the first
  vertex of label class~$i$ with color~$c^i_1$, color any~$j$th vertex,~$2\le j\le
  x-1$, with the list $\{c^i_{j-1},c^i_{j}\}$, and finally color
  the~$x$th vertex with color~$c^i_{x-1}$. Next, color each 
  vertex from $V_1\cup \{v^*\}$ with a unique color. Let~$\L$ denote
  the list-coloring of $V(G)$ that we just described. We define the motif~$M$ as
  the set containing each color present in $\L$.  Clearly, the reduction works
  in polynomial time. Note that~$|V|=kx+|E|+1$ and $|M|=k(x-1)+|E|+1$
  and thus~$\ell=|V|-|M|=k$. To prove our claim, it thus remains to
  show the correctness of the reduction.

  \begin{quote}
    $(H,k)$ is a yes-instance of \textsc{Multicolored Independent Set}
    $\Leftrightarrow$~$(M,G,\L)$ is a yes-instance of \sLGM.
  \end{quote}

  ($\Rightarrow$) Let~$S$ be a size-$k$ independent set with pairwise different vertex labels
  in~$H$. Consider the set~$Y:=V\setminus S$ in~$G$. First, note
  that~$G[Y]$ is connected: vertex~$v^*$ is adjacent to all vertices
  in~$Y\cap V_0$ and each vertex~$v_{\{u,w\}}$ of~$Y\cap V_1 =V_1$ has at least one neighbor in~$Y\cap V_0$, because at most one of the
  endpoints of~$\{u,w\}$ is in the independent set~$S$.
  
  It remains to show that we can assign colors to the vertices such
  that the union of the vertex colors is~$M$. All vertices
  with unique colors are contained in~$Y$ and their coloring is clear. All other vertices are
  in~$V_0$. Now consider label class~$i$ of~$V_0$. Exactly one vertex~$u$
  of label class~$i$ is contained in~$S$. Let~$j$ be the number such
  that~$u$ is the~$j$th vertex of the label class~$i$. Then, color
  the~$q$th vertex of label class~$i$ with color~$c^i_q$ if~$q<j$ and
  with color~$c^i_{q-1}$ if~$q>j$. Clearly this coloring assigns~$x-1$
  different colors to the vertices of each label class. Hence, there
  is a coloring of the vertices of~$Y$ that is equal to~$M$.

  ($\Leftarrow$) Let~$Y$ denote an occurrence of~$M$ in~$G$. First,
  observe that there are only~$x-1$ colors for the~$x$ vertices of
  each label class. Hence,~$Y$ contains exactly~$x-1$ vertices of each
  label class. Now let~$S$ denote the set containing, for each label
  class, the only vertex \emph{not} contained in~$Y$. Clearly, $|S|=k$
  and the elements of~$S$ have pairwise different labels in~$H$. Furthermore,~$S$ is an independent set
  in~$H$: since~$G[Y]$ is connected, there is for each edge
  vertex~$v_{\{u,w\}}$ at least one of its neighbors in~$Y$. Hence, at
  most one of the endpoints of each edge~$\{u,w\}$ is in~$S$.
\end{proof}

We immediately obtain the following.
\begin{corollary}\label{cor:lb-small-list}
  \LGM is W[1]-hard with respect to~$\ell$ even if~$|\L(v)|\le 2$ for
  every vertex $v$ in~$G$.
\end{corollary}

\section{Graph Motif on Trees}
\label{sec:gm}
Motivated by these negative results on general graphs, we now study the special case where the
input graph is a tree. For~\sLGM{}, we were not able to resolve the parameterized complexity
with respect to~$\ell$ for this case. Hence, we focus on the more restricted~\sGM{} problem. We
show that~\sGM{} is fixed-parameter tractable with respect to~$\ell$ if the input graph is a
tree. Recall that for general graphs,~\sGM{} is W[1]-hard for~$\ell$ even if the motif~$M$
contains only two colors~\cite{BBFKN11}. Hence, our result shows that the tree structure
significantly helps when parameterizing by~$\ell$. We then show that the fixed-parameter algorithm for
\sGM{} on trees extends to some special cases of \sLGM{} in which the vertex-color graph is
also a tree. Finally, we show that a polynomial-size kernel for~\sGM{} on trees
parameterized by~$\ell$ is unlikely.

\subsection{A Dynamic Programming Algorithm}  Call a
color of~$M$ \emph{abundant} if it occurs more often in~$G$ than
in~$M$. The abundant colors are exactly the ones that have to be
``deleted'' to obtain a solution~$S$. Let~$c_1,\ldots, c_j$ denote the
abundant colors of~$M$, and let~$\ell_i$ denote the difference between
the number of vertices in~$V$ that have color~$c_i$ and the
multiplicity $M(c_i)$ of~$c_i$ in~$M$. This implies in particular
that~$\sum_{1\le i\le j} \ell_i=\ell$.

The algorithm is a dynamic programming algorithm that works on a
rooted representation~$T$ of~$G$. We obtain~$T$ by choosing an arbitrary vertex~$r\in V$ and
rooting~$G$ at~$r$.  As usual for dynamic
programming on trees, the idea is to combine partial solutions of
subtrees. Our algorithm is somewhat similar to a previous dynamic
programming algorithm for \sGM{} on graphs of bounded
treewidth~\cite{FFHV11} but the analysis and concrete table setup is
different.

In the following, let~$T_v$ denote the subtree of~$T$ rooted at vertex~$v$. For each subtree,
we let~$\num(T_v,c)$ denote the number of vertices in~$T_v$ that have color~$c$. If a
solution contains vertices from~$T_v$ and further vertices, then it must contain~$v$ and all
vertices with nonabundant colors in~$T_v$.  Hence, in the dynamic programming it is sufficient
to consider subtrees described in the following definition.

\begin{definition}
We call a connected
subtree~$T'$ of~$T_v$ \emph{safe} if~$T'$ contains~$v$ and if every
vertex of~$T_v$ that is colored by a nonabundant color is contained
in~$T'$.
\end{definition}

We fill a family of dynamic programming tables~$D_v$, one table for each~$v\in V$. The entries of~$D_v$ are defined as
follows:
\[D_v[\lambda_1,\ldots,\lambda_j]=
 \begin{cases}
   1 & \text{if~$T_v$ has a safe subtree containing for each~$c_i$, $1\le i\le j$,}\\ & \text{exactly~$\num(T_v,c_i)-\lambda_i$ vertices of color~$c_i$}, \\
   0 & \text{otherwise}.
 \end{cases}
\]
Assume for now that the table has completely been filled out. Then, it
can be easily determined whether~$G$ has an occurrence~$S$ of~$M$.

If~$S$ is an occurrence of~$M$, then let~$v$ denote the root of~$T[S]$. Clearly,~$T[S]$ is a safe subtree
of~$T_v$. Moreover, every vertex with a nonabundant color is contained
in~$T[S]$ and for all vertices with an abundant color~$c_i$, the tree~$T[S]$
contains~$\num(T_v,c_i)-\lambda_i$ vertices with color~$c_i$ for
some~$\lambda_i\ge 0$. Thus, there is some table entry
$D_v[\lambda_1,\ldots,\lambda_j]$ whose value is~$1$ and
where~$\num(T_v,c_i)-\lambda_i$ is the multiplicity of~$c_i$ for
each~$c_i$.

Conversely, if there is some entry~$D_v[\lambda_1,\ldots,\lambda_j]$
with value 1 such that~$T_v$ contains all vertices with nonabundant
colors and for each~$c_i$, $1\le i\le j$,~$\num(T_v,c_i)-\lambda_i$ is
exactly the multiplicity of~$c_i$ in~$M$, then there is at least one
safe subtree of~$T_v$ whose vertex set is an occurrence of~$M$. 

Hence, one may solve~\sGM{} by filling table~$D$, and then checking
for each vertex~$v$ whether one of the entries
of~$D_v[\lambda_1,\ldots,\lambda_j]$ with value~1 implies the
existence of a solution. For the running time bound, the main
observation that we exploit is that if a safe rooted subtree of~$T_v$
contains all the vertices of~$T_v$ that are in a solution~$S$, then it
contains at least $\num(T_v,c_i)-\ell_i$ vertices with
color~$c_i$. Consequently, the relevant range of values
for~$\lambda_i$ is in~$[0,\ell_i]$ and thus bounded in the parameter
value~$\ell$. 

We now describe how to fill in table~$D$. To initialize~$D$,
consider each leaf~$v$ of the tree~$T$. By the 
definition of~$D$, an entry can have value~$1$ only if there is a
corresponding safe tree which needs to
contain~$v$. Thus, $$D_v[\lambda_1,\ldots,\lambda_j]=1 \Leftrightarrow
\lambda_1=\ldots = \lambda_j=0.$$ Now, to compute the entries of~$D$
for a nonleaf vertex~$v$, we combine the entries of the children
of~$v$. To this end, fix an arbitrary ordering of the children of~$v$
and denote them by~$u_1,\ldots ,u_{\deg(v)}$. Now, let~$T^i_v$ denote
the subtree rooted at~$v$ containing the vertices of
each~$T_{u_q}$,~$1\le q\le i$, and no vertices from
each~$T_{u_q}$,~$q>i$. For increasing~$i$, we compute solutions
for~$T^i_v$, eventually computing the solutions for~$T^{\deg(v)}_v=T_v$. To compute
these solutions, we define an auxiliary table~$D^i_v$. The table
entries are defined just as for~$D$, that is,
 \[D^i_v[\lambda_1,\ldots,\lambda_j]=
 \begin{cases}
   1 & \text{if~$T^i_v$ has a safe subtree containing for each~$c_i$,
     $1\le i\le j$,}\\ & \text{exactly~$\num(T_v,c_i)-\lambda_i$
     vertices of color~$c_i$}, \\ 
   0 & \text{otherwise}.
 \end{cases}
\]
Observe that, since $T^{\deg(v)}_v=T_v$, we have~$D^{\deg(v)}_v=D_v$ and thus by computing~$D^{\deg(v)}_v$ we also compute~$D_v$. 
Now,~$D^1_v$ can be computed in a straightforward fashion from the
entries of~$D_{u_1}$.
 \[D^1_v[\lambda_1,\ldots,\lambda_j]=
 \begin{cases}
   1 & \text{if~$D_{u_1}[\lambda_1,\ldots,\lambda_j]=1$}\\ 1 & \text{if~$\num(T_{u_1},c_i)=\lambda_i$ and~$T_{u_1}$ contains only abundant colors}, \\
   0 & \text{otherwise}.
 \end{cases}
\]
The first case corresponds to the case that the safe subtree~$T'$
of~$T^i_v$ contains at least one vertex of~$T_{u_1}$, the second case
corresponds to the case that~$T'$ contains only~$v$.

To compute~$D^i_v$ for~$i>1$, we combine entries of~$D^{i-1}_v$
with~$D_{u_i}$.
\begin{displaymath}
  D^i_v[\lambda_1,\ldots,\lambda_j] =
  \begin{cases}
    1 & \text{if $T_{u_i}$ contains only abundant colors and} \\ 
    &  D^{i-1}_v[\lambda_1-\num(T_{u_i},c_1),\ldots,\lambda_j-\num(T_{u_i},c_1)]=1,\\
    1 & \text{if there is $(\lambda'_1, \ldots , \lambda'_j)$ such that}\\
    &  D^{i-1}_v[\lambda_1',\ldots,\lambda_j']=D_{u_i}[\lambda_1-\lambda'_1,\ldots,\lambda_j-\lambda'_j]=1,\\
    0 & \text{otherwise.}
  \end{cases}
\end{displaymath}
The first case corresponds to the situation in which no vertex
of~$T_{u_i}$ is part of the safe subtree, in the second case the safe
subtree contains some vertices of~$T_{u_i}$ and some vertices
of~$T^{i-1}_v$. 

This completes the description of the dynamic programming
recurrences. The correctness follows from the fact that the
recurrence considers all possible cases to ``distribute'' the vertex
deletions. It remains to bound the running time.

\begin{theorem}
  \GM can be solved in~$\Oh(3^\ell\cdot n)$ time if~$G$ is a tree.
  \label{thm:gm-tree}
\end{theorem}
\begin{proof}
  As described above, the relevant values of each~$\lambda_i$ are in~$[0,\ell_i]$. Thus, for
  each subtable~$D^i_v$ and~$D_v$, the number of entries to compute
  is~$\prod_{1\le i\le j}(\ell_i+1)$. The dominating term in the overall running time is the
  computation of the second term in the recurrence for~$D^i_v[\lambda_1,\ldots,\lambda_j]$
  where we consider all $(\lambda'_1, \ldots , \lambda'_j)$ such that
  $D^{i-1}_v[\lambda_1',\ldots,\lambda_j']=
  D_{u_i}[\lambda_1-\lambda'_1,\ldots,\lambda_j-\lambda'_j]=1$. The number of possible choices
  can be computed as follows. For each~$\lambda_i$, the values of~$\lambda'_i$ range between~$0$ and~$\lambda_i$. Overall this gives

  $$\sum_{j=0}^{\ell_i} j+1 = \sum_{j=1}^{\ell_i+1} j= (\ell_i+2)\cdot(\ell_i+1)/2$$ possibilities for choosing~$\lambda_i$ and~$\lambda'_i$.
  We now bound this product in~$\ell$. Thus, we aim to find the
  vector~$(\ell_1,\ldots ,\ell_j)$ that maximizes $\prod_{1\le i\le j}(\ell_i+2)\cdot(\ell_i+1)/2$ under the
  constraint~$\sum_{1\le i\le j} \ell_i = \ell$. We claim that this is the vector
  with~$\ell_1= \ldots \ell_j=1$.

  To this end, consider a vector~$(\ell_1,\ldots ,\ell_j)$ whose
  maximum entry is at least~$2$. Without loss of generality, assume
  thus~$\ell_j>1$. Now consider~$(\ell_1,\ldots ,\ell_j-1,1)$ and
  observe that~$(\sum_{1\le i< {j}} \ell_i)+(\ell_j-1) + 1 = \ell$, that
  is, the new vector also satisfies the summation
  constraint. Moreover,
 $$\frac{\prod_{1\le i\le
     j}((\ell_i+2)\cdot(\ell_i+1)/2)}{(\prod_{1\le i< j}((\ell_i+2)\cdot(\ell_i+1)/2) \cdot ((\ell_j+1)(\ell_j)/2) \cdot
   3}=\frac{(\ell_j+2)(\ell_j+1)/2}{3(\ell_j+1)(\ell_j)/2}=\frac{\ell_j+2}{3(\ell_j)}<1,$$ where the inequality
 follows from~$\ell_j>1$. Since the new vector has more entries with
 value 1, we conclude that the maximum value is reached when all
 entries assume value 1. Consequently, the worst case number of recurrences that need to be evaluated for filling a subtable~$D^i_v$ or~$D_v$ is~$\Oh(3^\ell)$. 
 The overall number of subtables to fill
 is~$\Oh(\sum_{v\in V}\deg(v))=\Oh(n)$. This implies the overall running
 time bound.
\end{proof}

\subsection{An Extension to Subcases of List-Colored Graph Motif on Trees}
The fixed-parameter tractability of~\sGM{} on trees can be extended to give fixed-parameter tractability for \sLGM{} when the input graph~$G$ is a tree and the vertex-color graph~$H$ is a forest with bounded degree.

The first step in our algorithm is to apply the following two data reduction rules which are obviously correct.
\begin{rrule}\label{rule:bad-color}
  If there is a color vertex~$v$ in~$H$ such that the degree of~$v$ in~$H$ is smaller than the multiplicity of~$v$ in~$M$, then return ``no''. 
\end{rrule}

\begin{rrule}\label{rule:tight-color}
  If there is a color vertex~$v$ in~$H$ such that the degree of~$v$ in~$H$ equals the multiplicity of~$v$ in~$M$, then set~$\L(u)=\{v\}$ for all neighbors~$u$ of~$v$ in~$H$. 
\end{rrule}

With these reduction rules at hand, we can show that the following special case of \sLGM{} is
fixed-parameter tractable with respect to~$\ell$.

\begin{lemma}\label{cor:sgm-to-gm}
  \sLGM{} can be solved in~$\Oh(3^\ell\cdot n)$ time if~$G$ is a tree and the vertex-color
  graph~$H$ is a forest in which for every color
  vertex~$c$, the difference between the degree of~$c$ in~$H$ and the multiplicity of~$c$
  in~$M$ is at most one.
\end{lemma}
\begin{proof}
  We describe a reduction of this special case of \sLGM{} on trees to~\sGM{} on trees. In the
  following, we assume without loss of generality that every color in the instance has
  multiplicity at least one in~$M$.
  First, apply Reduction Rules~\ref{rule:bad-color}
  and~\ref{rule:tight-color} exhaustively. This can be done in~$\Oh(n)$ time by computing the difference between~$M(c)$ and~$\deg_H(c)$ once for all~$c\in C$ and updating this value whenever we delete an edge. Afterwards, for every color
  vertex~$c$ in~$M$, we have~$M(c)\le \deg_H(c)$, due to Reduction Rule~\ref{rule:bad-color}, and~$M(c)\ge \deg_H(c)-1$ since the
  reduction rules do not increase the degree of a vertex. Moreover, if~$M(c)=\deg_H(c)$ in~$M$, then the connected component of~$c$ in~$H$ consists of~$c$
  and~$\deg_H(c)$ leaf neighbors of~$c$. By the above, the vertex-color graph~$H$ contains the
  two types of connected components considered below. For both of them, we show that the
  constraints of~$\L$ can be replaced by simple coloring constraints. Consider a connected
  component~$H'$ of~$H$.

  \emph{Case~1: $H'$ is a star with a central color vertex~$c$ such that~$M(c)=\deg_H(c)$
    in~$M$.}  Replace~$c$ by~$\deg_H(c)$ new colors and assign a different color to each
  neighbor of~$c$ in~$H$. This is equivalent as all neighbors of~$c$ in~$H$ are
  contained in any occurrence of~$M$.

  \emph{Case~2: $H'$ is a tree in which each color vertex~$c$ fulfills~$M(c)=\deg_H(c)-1$.}
  Let~$C'$ denote the set of color vertices in~$H'$ and~$V'$ denote the set of vertices of~$H'$
  that do not correspond to colors. Replace~$C'$ by one new color~$c^*$ and set the
  multiplicity of~$c^*$ to~$|V'|-1$. To show correctness of this replacement, we show that for
  every~$v\in V'$, there is an assignment~$f':V'\setminus \{v\}\to C'$ such
  that~$f'(u)\in \L(u)$ for each $u\in V'\setminus \{v\}$ and each color~$c\in C'$ is assigned
  exactly~$\deg_H(c)-1$ vertices by~$f'$. To see the existence of this assignment consider the
  version of~$H'$ that is rooted at~$v$. For each color vertex~$c$ in~$H'$, let~$V'_c$ denote
  the children of~$c$ in this rooted tree. For each vertex~$u\in V'_c$ set~$f'(u):=c$. With
  this assignment, there are exactly~$\deg_H(c)-1$ vertices that are assigned to~$c$ and every
  vertex~$u\in V'\setminus \{v'\}$ is assigned to some color~$c$ of~$C'$, namely to its
  predecessor in~$H'$.
  
  Applying these replacements exhaustively then results in an
  equivalent instance of~\sGM{} on trees which can be solved in the claimed
  running time due to Theorem~\ref{thm:gm-tree}.
\end{proof}

We now show how to use the running time bound of Lemma~\ref{cor:sgm-to-gm} to obtain a
fixed-parameter algorithm for the dual parameter~$\ell$ for the special case of \sLGM{}
when the color-vertex graph is a tree and each color has a bounded multiplicity
in~$M$. Thus, let~$M(C):=\max_{c\in C} \mult(c)$ denote the largest multiplicity in~$M$. We will achieve the
algorithm by branching on colors~$c$ where the difference between~$M(c)$ and~$\deg_H(c)$ is at least two. We call such a color vertex
\emph{2-abundant} in the following. The first step of the algorithm is to apply Reduction
Rules~\ref{rule:bad-color} and~\ref{rule:tight-color} exhaustively.

\begin{brule}\label{brule:abundant-in-tree}
  If the vertex-color graph~$H$ contains a connected component~$H'$ with at least one 2-abundant
  color vertex, then do the following.
  \begin{itemize}
  \item Root~$H'$ arbitrarily. 
  \item Choose some 2-abundant vertex~$c$
    of~$H'$ such that the subtree of~$H'$ rooted at~$c$ has no further 2-abundant~vertex. 
  \item Choose a set~$V_c$ of~$\mult(c)+1$ arbitrary children of~$c$. 
  \item For each~$u\in V_c$ branch into the case that~$c$ is removed from~$\L(u)$.
\end{itemize}
\end{brule}
\begin{proof}[Proof of correctness.]
First, observe that such a 2-abundant color vertex~$c$ always exists and that it
can be found in linear time by a bottom-up traversal of the rooted tree. Second, observe
that since~$\mult(c)\ge \deg_H(c)+2$, the vertex~$c$ has at least~$\mult(c)+1$
children. Hence, some child~$u$ of~$c$ in~$H'$ does not receive the color~$c$ in any
occurrence of~$M$. Thus, if the original instance is a yes-instance, then the branch in
which we remove~$c$ from~$\L(u)$ is a yes-instance. Conversely, any occurrence of~$M$ in
an instance created by the branching rule is an occurrence of~$M$ in the original
instance.
\end{proof}
If Branching Rule~\ref{brule:abundant-in-tree} does not apply, then we can solve the instance
in~$\Oh(3^\ell\cdot n)$ time by Lemma~\ref{cor:sgm-to-gm}.  It thus remains to ensure that the rule
cannot be applied too often. To this end, we apply one further reduction rule. To formulate the
rule we need the following definition. We call a connected component~$H'$ of the vertex-color
graph~$H$ \emph{costly} if~$H'$ either consists of just one vertex~$v\in V$ or~$H'$ is a tree
such that all color vertices~$c$ in~$H'$ have multiplicity exactly~$\deg_H(c)-1$ in~$M$.
\begin{rrule}\label{rule:costly-comps}
  If~$G$ contains at least~$\ell+1$ costly components, then return ``no''.
\end{rrule}
\begin{proof}[Proof of correctness.]
For
each costly component at least one vertex is not contained in any occurrence of~$M$. This
is obvious for those components consisting only of one vertex~$v$ from~$V$. For the other
costly components, this follows from Case~2 in the proof of Lemma~\ref{cor:sgm-to-gm}.
\end{proof}
Now it remains to observe that in each instance created by an application of Branching
Rule~\ref{brule:abundant-in-tree}, the number of costly components is increased by exactly
one. Hence, after at most~$\ell+1$ branching steps, Reduction Rule~\ref{rule:costly-comps} directly
reports that we have a ``no'' instance. Since we branch into~$\mult(c)+1\le \mult_C+1$ cases in
each application of Branching Rule~\ref{brule:abundant-in-tree}, we thus
create~$\Oh((\mult_C+1)^{\ell+1}))$ instances that either adhere to the conditions of
Lemma~\ref{cor:sgm-to-gm} or are rejected due to Reduction Rules~\ref{rule:bad-color} or~\ref{rule:costly-comps} and can thus be solved in~$\Oh(3^\ell\cdot n)$ time. Altogether, we obtain
the following running time.
\begin{theorem}\label{thm:lgm-tree-H-bounded-mult}
  If~$G$ is a tree, and the color vertex graph~$H$ is a forest, then~\sLGM{} can be solved in
  $\Oh((\mult_C+1)^{(\ell+1)}\cdot 3^\ell\cdot n)$~time.
\end{theorem}
When~$M$ is a set, the largest multiplicity is one, giving the following running
time.
\begin{corollary}
  If~$G$ is a tree,~$H$ is a forest, and~$M$ is a set, then~\sLGM{} can be solved in
  $\Oh(6^\ell\cdot n)$~ time.
\end{corollary}
By observing that Branching Rule~\ref{brule:abundant-in-tree} branches into at
most~$\deg_H(c)-1$ branches, we also obtain the following running time bound in terms of
the maximum degree of color vertices in~$H$.
\begin{corollary}
  If~$G$ is a tree, and~$H$ is a tree whose color vertices have degree at most~$\Delta_C$, then~\sLGM{} can be solved in $\Oh((\Delta_C-1)^{(\ell+1)}\cdot 3^\ell\cdot n)$~time.
\end{corollary}

\subsection{A Kernelization Lower Bound}  We now show that \sGM does
not admit a polynomial-size problem kernel with respect to~$\ell$, even
if~$G$ is a tree. The proof is based on a
cross-composition~\cite{BJK14} from the \MCC problem.

\begin{quote}
  \MCC\\
  \textbf{Input:} An undirected graph~$H=(W,F)$ and a vertex-labeling~$\lambda:W\to \{1,\ldots ,k\}$.\\
  \textbf{Question:} Is there a vertex set~$S\subseteq W$
  such that~$|S|=k$, the vertices in~$S$
  have pairwise different labels and~$H[S]$ is a clique? 
\end{quote}

\MCC has been shown to be W[1]-hard parameterized by $k$~\cite{FHRV09}.
We refer to the colors of the \MCC instance as labels to avoid confusion with the colors of the \sGM
instance. 
Informally, cross-compositions are reductions that combine many
instances of one problem into one instance of another
 problem. The existence of a cross-composition from an NP-hard
problem to a parameterized problem~$Q$ implies that~$Q$ does not admit
a polynomial-size problem kernel (unless~$\text{NP}\subseteq
\text{{coNP/poly}}$)~\cite{BJK14}. 
\begin{definition}[\cite{BJK14}]
  Let $L\subseteq \Sigma^*$ be a language, let~$R$ be a polynomial
  equivalence relation on $\Sigma^*$, and let $Q\subseteq
  \Sigma^*\times \mathbb{N}$ be a parameterized problem.  An
  \emph{or-cross-composition of $L$ into $Q$} (with respect to~$R$) is
  an algorithm that, given $t$ instances $x_1, x_2, \ldots , x_t \in
  \Sigma^*$ of $L$ belonging to the same equivalence class of $R$,
  takes time polynomial in~$\sum_{i=1}^t |x_i| + k$ and outputs an
  instance $(y,k)\in \Sigma^*\times \mathbb{N}$ of $Q$ such that
\begin{itemize}
\item the parameter value $k$ is polynomially bounded in $\max_{i=1}^t | x_i | +\log t$, and
\item the instance $(y,k)$ is a yes-instance for $Q$ if
  and only if at least one instance $x_i$ is a yes-instance for $L$.
\end{itemize}
\end{definition}

We present an or-cross composition of~\MCC into \sGM on trees
parameterized by~$\ell$. The polynomial equivalence relation~$R$ will
be simply to assume that all the \MCC instances have the same number
of vertices~$n$. The main trick is to encode vertex identities in the
graph of the~\MCC instance by numbers of colored vertices in the \sGM
instance; this approach was also followed in previous works
on~\sGM~\cite{FFHV11,BS15}.

Given~$t$
instances~$(H_1=(W_1,F_1),\lambda_1), H_2=(W_2,F_2),\lambda_2), \ldots ,
H_t=(W_t,F_t),\lambda_t)$ of~\MCC such that~$|W_i|=n$ for all~$i\in [t]$,
we reduce to an instance of \sGM where the input graph is a tree as
follows. Herein, we assume without loss of generality that~$t=2^s$ for
some integer~$s$.

The first construction step is to add one vertex~$r$
that connects the different parts of the instance and which will be
contained in every occurrence of the motif. The vertex~$r$ thus
receives a unique color that may not be deleted. To this vertex~$r$ we
 attach subtrees corresponding to edges of the input
instances. Deleting vertices of such a subtree then corresponds to
selecting the endpoints of the corresponding edge. 

\subparagraph*{Instance selection gadget.} The technical difficulty in the
construction is to ensure that the solution of \sGM deletes only vertices in
subtrees corresponding to edges of the same graph. To achieve this, we
introduce~$k\cdot (k-1)\cdot \log t$ instance selection
colors~$\iota[p,q,\tau]$ where~$p\in [k]$,~$q\in [k]\setminus \{p\}$,
and~$\tau\in [\log t]$, and demand that the solution deletes exactly
one vertex of each instance selection color. To ensure that exactly
one instance is selected, we use two further colors~$\iota^+$
and~$\iota^-$.

For each \MCC instance~$(H_i,\lambda_i)$, attach an \emph{instance
  selection path~$P_i$} to~$r$ that is constructed based on the
number~$i$. Let~$b(i)$ 
denote the binary expansion of~$i$ and let~$b_\tau(i)$,~$\tau\in [\log t]$,
denote the~$\tau$th digit of~$b(i)$. Construct a
path~$P_i$ containing first a vertex with color~$\iota^+$, then in
arbitrary order exactly one vertex of each color in the color
set $I_i:=\{\iota[p,q,\tau] : b_\tau(i)=1\}$, and then a vertex
with color~$\iota^-$. Attach the path~$P_i$ to~$r$ by making the vertex
with color~$\iota^+$ a neighbor of~$r$.

The idea of the construction is that exactly one instance selection
path~$P_i$ is completely deleted and that this will force any solution to
delete paths that ``complement''~$P_i$ (that is, paths which
contain all $\iota[p,q,\tau]$ such that $b_\tau(i)=0$) in the rest of
the graph.

\subparagraph*{Edge selection gadget.} To force deletion of subtrees
corresponding to exactly~$\binom{k}{2}$ edges with different labels,
we introduce~$2k(k-1)$ label selection colors~$\lambda[p,q]^+$
and~$\lambda[p,q]^-$ where~$p\in [k]$ and~$q\in [k]\setminus
\{p\}$. These colors will ensure that, for each pair of labels~$p$
and~$q$, the solution deletes exactly one path corresponding to the
ordered pair~$(p,q)$ and one path corresponding to the
pair~$(q,p)$. 

There are two further sets of colors. One set is used for ensuring
vertex consistency of the chosen edges, that is, to make sure that all
the selected edges with label pair~$(p,\cdot)$ correspond to the same
vertex with label~$p$. More precisely, we introduce a
color~$\omega[p,q]$ for each~$p\in [k]$ and each~$q\in [k]\setminus
\{p\}$, except for the biggest~$q\in [k]\setminus \{p\}$.

The final color set is used to check that the edges selected for label
pair~$(p,q)$ and for label pair~$(q,p)$ are the same. To this end, we introduce a set
of colors~$\varepsilon[p,q]$ for each~$p\in [k]$ and each~$q\in
[k]\setminus \{p\}$ such that~$q>p$. To perform the checks of vertex
and edge consistency, we encode the identities of vertices and edges
into path lengths. More precisely, we assign each vertex~$v\in W_i$ a unique
(with respect to the vertices of~$W_i$) number~$\#(v)\in [n]$. 

Now, for each label pair~$(p,q)$ and each instance~$i$, attach one
path~$P_i(u,v)$ to~$r$ for each edge~$\{u,v\}$ where~$u$ has color~$p$
and~$v$ has color~$q\neq p$. The path~$P_i(u,v)$
\begin{itemize}
\item starts with a vertex with color~$\lambda[p,q]^+$ that is made adjacent to~$r$,
\item then contains exactly one vertex of each color in~$\{\iota[p,q,\tau]: \iota[p,q,\tau]\notin
  I_i\}$,
\item then contains~$\#(u)$ vertices of color~$\varepsilon[p,q]$
  if~$p<q$ and~$n-\#(v)$ vertices of color~$\varepsilon[q,p]$ if~$p>q$,
\item then, if~$q$ is not the biggest label in~$[k]\setminus p$,
  contains~$\#(u)$ vertices with color~$\omega[p,q]$,  
\item then, if~$q$ is not the smallest label in~$[k]\setminus p$,
  contains~$n-\#(u)$ vertices with color~$\omega[p,q']$, where~$q'$ is
  the next-smaller label in~$[k]\setminus p$ (if~$p=q-1$,
  then~$q'=q-2$; otherwise~$q'=q-1$), and
\item ends with a vertex with color~$\lambda[p,q]^-$.
\end{itemize}

Let~${\cal C}$ denote the multiset containing all the vertex colors of all
vertices added during the construction with their respective
multiplicities. In the correctness proof it will be easier to argue
about the colors that are {\em not} contained in~$M$. Hence, the
construction is completed by setting the multiset~$D$ of colors to
``delete'' to contain each color exactly once except
  
\begin{itemize}
\item the color of~$r$ which is not contained in~$D$,
\item the vertex consistency colors~$\omega[p,q]$ each of which is
  contained with multiplicity~$n$, and
\item the edge selection colors $\varepsilon[p,q]$ each of which is contained with multiplicity~$n$.
\end{itemize}
The motif~$M$ is defined as~$M:= {\cal C}\setminus D$. 
It remains to show the correctness.
\begin{theorem}\label{thm:gm-tree-no-kernel}
  \GM does not admit a polynomial-size problem kernel with respect
  to~$\ell$ even if~$G$ is a tree unless~$\text{\rm NP}\subseteq \text{\rm {coNP/poly}}$.
\end{theorem}
\begin{proof}
  To complete the proof we need to show that the construction fulfills
  the properties of cross-compositions. First, the construction
  clearly runs in polynomial time. Second, the number of introduced
  colors is polynomial in~$k +\log t$ and thus the value of~$\ell=|D|$
  is bounded polynomial in~$n+\log t$. Thus, it remains to show that
  the composition is an or-cross composition, that is:
  \begin{quote}
    At least one~$(H_i,\lambda_i)$ is a yes-instance of~\MCC
    $\Leftrightarrow$ $(M,G,\L)$ is a yes-instance of~\sGM.
  \end{quote}

  ($\Rightarrow$) Let~$S\in W_i$ be a vertex set of size~$k$ such
  that~$H_i[S]$ is a clique and the vertices in~$S$ have pairwise different labels. Consider the induced subgraph~$G'$ of~$G$
  obtained by completely deleting the path~$P_i$ and, for
  each~$\{u,v\}\in H_i[S]$, the paths~$P_i(u,v)$ and~$P_i(v,u)$. Since
  only complete paths are deleted and since each path in~$G$ is
  attached to~$r$, the graph~$G'$ is connected. It remains to show
  that the multiset of deleted colors is~$D$. First,~$\iota^+$
  and~$\iota^-$ are deleted once and contained once in~$D$. Second,
  each instance selection color~$\iota[p,q,\tau]$ is deleted once as
  required by~$D$: If $\iota[p,q,\tau]$ is contained in~$P_i$, then it
  is not contained in any~$P_i(u,v)$. Conversely, if $\iota[p,q,\tau]$
  is not contained in~$P_i$ then it is contained in each~$P_i(u,v)$
  where~$u$ has color~$p$ and~$v$ has color~$q$. Third, exactly~$n$
  vertices of each vertex consistency color~$\omega[p,q]$ are deleted:
  these vertices are contained only in two paths~$P_i(u,v)$, namely
  if~$u$ has label~$p$ and~$v$ has either label~$q$ or
  label~$q+1$. Since all the deleted paths with label pair~$(p,\cdot)$
  correspond to the same vertex~$u$, the number of vertices with color
  $\omega[p,q]$ is~$\#(v)$ if~$v$ has label~$q$ and~$n-\#(v)$ if~$v$
  has label~$p$. Hence, exactly~$n$ vertices with this color are
  deleted, as required by~$D$. Finally, we show that exactly~$n$
  vertices of each edge selection color~$\varepsilon[p,q]$,~$p<q$, are
  deleted: Let~$u$ and~$v$ be the vertices of~$S$ with label~$p$
  and~$q$, respectively. Then, the deleted path~$P_i(u,v)$
  contains~$\#(u)$ vertices with color~$\varepsilon[p,q]$ and the
  deleted path~$P_i(v,u)$ contains~$n-\#(u)$ vertices with this
  color. Altogether, the
  multiset of colors in~$G'$ is exactly~${\cal C}\setminus D=M$.

  ($\Leftarrow$) Let~$G'$ be a connected subgraph of~$G$ whose multiset
  of vertex colors is exactly~$M$. Let~$V_D:= V(G)\setminus V(G')$
  denote the set of deleted vertices, that is, vertices not
  in~$G'$. The color multiset of the vertex colors of~$V_D$ is
  exactly~$D$. Thus, exactly one vertex with color~$\iota^+$ and
  color~$\iota^-$ is deleted. Consequently exactly one path~$P_i$ is
  completely deleted from~$G'$: deleting~$\iota^+$ in some~$P_i$
  implies that the~$\iota^-$ in~$P_i$ is also deleted. Thus, no
  further vertices from any~$P_j$,~$j\neq i$, may be deleted. 
  
  Moreover, since each label selection color~$\lambda[p,q]^+$
  or~$\lambda[p,q]^-$ is contained exactly once in~$D$, the set~$V_D$
  also contains exactly one path~$P_j(u,v)$ where~$u$ has label~$p$
  and~$v$ has label~$q$. Moreover, we have~$j=i$ by the assignment of
  the instance selection colors: If~$j\neq i$, then there is
  some~$\tau\in [\log t]$ such that~$b_\tau(i)\neq b_\tau(j)$. Then,
  however~$\iota[p,q,\tau]$ is either not contained in the colors
  of~$V_D$ or it is contained twice in the colors of~$V_D$. In either case,
  the set of deleted colors is different from~$D$. 

  Thus, all the deleted paths in the edge selection gadgets correspond
  to the same instance~$i$. Now consider the paths for label
  pairs~$(p,\cdot)$. These label pairs correspond to the same vertex:
  Otherwise, there is some~$P_i(u,v)$ and some~$P_i(u',v')$ such
  that~$u\neq u'$,~$v$ has label~$q$, and~$v'$ has label~$q+1$. Then,
  however, the number of vertices with color~$\omega[p,q]$ does not
  equal~$n$ since~$P_i(u,v)$ contains~$\#(u)$ vertices of this
  color,~$P_i(u',v')$ contains~$n-\#(u')$ vertices of this color
  and~$\#(u)\neq \#(u')$. Hence, the deleted paths correspond to a
  vertex set~$S$ with~$k$ different labels in some~$H_i$. It remains
  to show that the graph~$H_i[S]$ is a clique.

  Consider an arbitrary pair of labels~$p$ and~$q$ where~$p<q$. Moreover,
  let~$u\in S$ and~$v\in S$ have label~$p$ and~$q$,
  respectively. Let~$P_i(u,v')$ be the path for~$u$ that is deleted
  with this color pair and let~$P_i(v,u')$ be the path for~$v$ that is
  deleted for this color pair. Then, $P_i(u,v')$~contains
  exactly~$\#(u)$ vertices with
  color~$\varepsilon[p,q]$,~$P_i(v,u')$~contains exactly~$n-\#(u')$
  vertices of this color. Since~$D$ contains exactly~$n$ vertices of
  this color, this implies~$u=u'$. By construction, this implies
  that~$u$ and~$v$ are neighbors in~$H_i$.
\end{proof}

\section{Colorful Graph Motif on Trees}
\label{sec:cgm}
For the combination of vertex-colored trees as input graphs and motifs that are sets, the
problem becomes considerably easier.  First, we show that in this case~\sCGM admits a
linear-vertex problem kernel that can be computed in linear time. The idea for the problem
kernelization is based on two simple observations.  First, we observe that the number of
vertices that are not unique is bounded in~\sCGM.
\begin{lemma}
  \label{lem:dual-bound}
  Let~$(M,G,\chi)$ be an instance of~\CGM. Then at most~$2\ell$
  vertices in~$G$ are not unique.
\end{lemma}
\begin{proof}
  Let~$C^+$ denote the set of colors that occur more than once in~$G$ and
  let~$\occ(c)$ denote the number of occurrences of a color~$c$
  in~$G$. We denote $c^+:=|C^+|$, $n^+:=\sum_{c\in C^+}\occ(c)$, and
  $n^-$ the number of unique vertices in $G$. By definition, no color is
  repeated in~$M$, thus $|M|= c^++n^-$~; moreover, $|V|=n^++n^-$. Hence, the 
  number $\ell=|V|-|M|$ of vertices to delete satisfies $\ell=
  n^+-c^+$. By definition $n^+\geq 2c^+$, and thus we conclude that $\ell\geq n^+/2$.
\end{proof}
Second, if there are two vertices that are unique, then the uniquely
determined path between these vertices is contained in every
occurrence of the motif. The kernelization accordingly removes all the
vertices that lie on these paths. More precisely, these vertices are
``contracted'' into the root~$r$. Afterwards, in a second phase some
further vertices are removed because their colors have been used
during the contraction. Eventually, this results in an instance which
has at most one unique vertex and thus, by Lemma~\ref{lem:dual-bound},
bounded size. For an example of the kernelization,
see Figure~\ref{fig:kernel}. Below, we give a more detailed description.

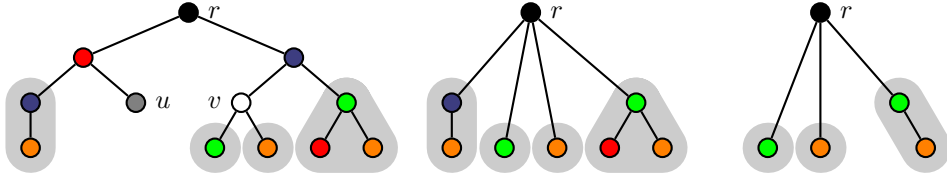
\begin{figure}[t]
  \centering
  \begin{tikzpicture}[xscale=1.4,yscale=0.6]

    \tikzstyle{vertex}=[circle,draw,thick,fill=black,minimum
    size=7pt,inner sep=1pt,font=\footnotesize]
    \tikzstyle{vertexb}=[circle,draw,thick,fill=dblue,minimum
    size=7pt,inner sep=1pt,font=\footnotesize]
    \tikzstyle{vertexr}=[circle,draw,thick,fill=red,minimum
    size=7pt,inner sep=1pt,font=\footnotesize]
    \tikzstyle{vertexg}=[circle,draw,thick,fill=green,minimum
    size=7pt,inner sep=1pt,font=\footnotesize]
    \tikzstyle{vertexgr}=[circle,draw,thick,fill=grey,minimum
    size=7pt,inner sep=1pt,font=\footnotesize]
    \tikzstyle{vertexw}=[circle,draw,thick,minimum
    size=7pt,inner sep=1pt,font=\footnotesize]
    \tikzstyle{vertexo}=[circle,draw,thick,fill=orange,minimum
    size=7pt,inner sep=1pt,font=\footnotesize]

    \tikzstyle{edge}=[-,thick]
    \tikzstyle{bedge} = [color=black,opacity=.2,line cap=round, line
      join=round, line width=19pt]
      
    \begin{scope}
      \node[vertex,label=right:$r$] (1) at (0,10) {}; \node[vertexr]
      (2) at (-1,9) {}; \node[vertexb] (3) at (1,9) {}; \node[vertexb]
      (4) at (-1.5,8) {}; \node[vertexgr,label=right:$u$] (5) at (-0.5,8) {};
      \node[vertexw,label=left:$v$] (6) at (0.5,8) {}; \node[vertexg] (7) at (1.5,8) {}; \node[vertexg] (8) at (0.25,7) {}; \node[vertexo] (9) at
      (0.75,7) {}; \node[vertexr] (10) at (1.25,7) {}; \node[vertexo]
      (11) at (1.75,7) {};
      \node[vertexo] (12) at (-1.5,7) {};

      \draw[edge] (2)--(1)--(3); \draw[edge] (12)--(4)--(2)--(5);
      \draw[edge] (6)--(3)--(7); \draw[edge] (8)--(6)--(9);
      \draw[edge] (10)--(7)--(11);
      \begin{pgfonlayer}{background}
         \draw[bedge]
        (4.center)--(12.center);
         \draw[bedge]
        (8.center)--(8.center);
         \draw[bedge]
        (9.center)--(9.center);
        \draw[bedge]
        (7.center)--(10.center)--(11.center)--(7.center);
      \end{pgfonlayer}
    \end{scope}
    \begin{scope}[xshift=3cm]
            \node[vertex,label=right:$r$] (1) at (0.25,10) {};  \node[vertexb]
      (4) at (-0.5,8) {}; 
       \node[vertexg] (7) at (1.25,8)
      {}; \node[vertexg] (8) at (0,7) {}; \node[vertexo] (9) at
      (0.5,7) {}; \node[vertexr] (10) at (1,7) {}; \node[vertexo]
      (11) at (1.5,7) {};
      \node[vertexo] (12) at (-0.5,7) {};

      \draw[edge] (1)--(7); \draw[edge] (12)--(4)--(1);
      \draw[edge] (8)--(1)--(9);
      \draw[edge] (10)--(7)--(11);
      \begin{pgfonlayer}{background}
         \draw[bedge]
        (4.center)--(12.center);
         \draw[bedge]
        (8.center)--(8.center);
         \draw[bedge]
        (9.center)--(9.center);
        \draw[bedge]
        (7.center)--(10.center)--(11.center)--(7.center);
      \end{pgfonlayer}
    \end{scope}
    \begin{scope}[xshift=5.5cm]
            \node[vertex,label=right:$r$] (1) at (0.5,10) {};  
       \node[vertexg] (7) at (1.25,8)
      {}; \node[vertexg] (8) at (0,7) {}; \node[vertexo] (9) at
      (0.5,7) {}; \node[vertexo]
      (11) at (1.5,7) {};

      \draw[edge] (1)--(7); 
      \draw[edge] (8)--(1)--(9);
      \draw[edge] (7)--(11);
      \begin{pgfonlayer}{background}
         \draw[bedge]
        (8.center)--(8.center);
         \draw[bedge]
        (9.center)--(9.center);
        \draw[bedge]
        (11.center)--(7.center);
      \end{pgfonlayer}
    \end{scope}
  \end{tikzpicture}

  \caption{The two phases of the kernelization. Left: the input
    instance, where~$r$,~$u$,~and~$v$ have unique colors; the pendant
    non-unique subtrees are highlighted by the grey 
    background. Middle: after Phase~I, all vertices on paths between
    unique vertices are contracted into~$r$. Right: in Phase~II, all
    vertices with a color that was removed in Phase~I are removed
    together with their descendants.}
\label{fig:kernel}
\end{figure}
\begin{theorem}
  \label{thm:cgm-tree-kernel}
  \CGM on trees admits a problem kernel with at most~$2\ell+1$
  vertices that can be computed in~$\Oh(n)$ time.
\end{theorem}
\begin{proof}
  We first describe the kernelization algorithm, then we show its
  correctness and finally bound its running time. By
  Lemma~\ref{lem:dual-bound}, the size bound holds if the instance has no
  unique vertex. Thus, we assume that there is a unique vertex in the
  following.

  Given an instance~$(M,G,\chi)$ of \sCGM, first root the input
  tree~$G$ at an arbitrary unique vertex~$r$. Now call a subtree with
  root~$v$ \emph{pendant} if it contains all descendants of~$v$
  in~$G$. Then, compute in a bottom-up fashion maximal pendant
  subtrees such that no vertex in this subtree is unique. Call these
  subtrees the \emph{pendant non-unique subtrees}. By
  Lemma~\ref{lem:dual-bound}, the total number of vertices in pendant
  non-unique subtrees is at most~$2\ell$. Now the algorithm removes
  vertices in two phases.

  \emph{Phase~I.} Remove from~$G$ all vertices except~$r$ that are not
  contained in a pendant non-unique subtree. Remove all colors of
  removed vertices from~$M$. If there is a color~$c$ such that two
  vertices with color~$c$ are removed in this step, then return
  ``no''. Make~$r$ adjacent to the root of each pendant non-unique
  subtree.

  \emph{Phase~II.} In the first step of this phase, for each color~$c$
  where at least one vertex has been removed in Phase~I, remove all
  vertices from~$G$ that have color~$c$. In the second step of this
  phase, remove all descendants of these vertices. Finally, let~$M'$
  denote the set of colors that are contained in the remaining
  instance.
  This completes the kernelization algorithm;
  the resulting instance has at most~$2\ell+1$ vertices since all
  vertices except~$r$ are unique.
  To show correctness, we first observe the following. 

  \emph{Claim: every occurrence of~$M$ in~$G$ contains no vertex~$v$
    that is removed during Phase~II of the kernelization.} This can be
  seen as follows. First, every occurrence of~$M$ in~$G$ contains all
  vertices removed during Phase~I: these vertices are either unique or
  lie on the uniquely determined path between two unique vertices. Now
  consider a vertex~$v$ removed during Phase~II. If~$v$ is removed in
  the first step of Phase~II, then $v$ has the same color~$c$ as a
  vertex~$u$ removed during Phase~I. Consequently, $v$ is not contained
  in an occurrence of~$M$: By the above, the occurrence contains~$u$
  and it contains no other vertex with color~$c$. Otherwise,~$v$ is
  removed in the second step of Phase~II, because~$v$ is not connected
  to~$r$. Since every occurrence of~$M$ contains~$r$, it thus cannot
  contain~$v$.

  We now show the correctness of the kernelization, that is, the
  equivalence of the original instance~$(M,G,\chi)$ and the resulting
  instance~$(M',G',\chi')$.
  First, assume that~$(M,G,\chi)$ is a yes-instance. Let~$S_T$ be an
  occurrence of~$M$ in~$G$, and let $T$ denote $G[S_T]$; by the above
  claim, $T$~contains only vertices that are removed during Phase~I or
  that are contained in~$G'$.  Consider the subtree~$T'$ of~$G$ that
  contains all vertices of~$T$ that are not removed during the
  kernelization. We show that~$T'$ is connected in~$G'$ and contains
  all colors of~$M'$. Connectivity can be seen as follows. First,
  observe that~$T$ and~$T'$ contain~$r$. Second, any vertex~$v\neq r$
  of~$T'$ is contained in some pendant non-unique subtree
  of~$G$. Thus,~$v$ is in~$T$ connected to~$r$ via a path that first
  visits only vertices of~$T'$, including the root of the pendant
  non-unique subtree. The root of the pendant non-unique subtree is
  in~$G'$ adjacent to~$r$. Thus, each vertex~$v\neq r$ has in~$T'$ a
  path to~$r$ which implies that~$T'$ is connected. It remains to
  prove that~$T'$ contains all colors of~$M'$. Consider a color~$c\in
  M'$. Since~$c\in M'$, none of the vertices with color~$c$ are
  removed in Phase~I of the kernelization. Moreover, since no vertex
  of~$T$ is removed in Phase~II of the kernelization, we have that the
  vertex of~$T$ with color~$c$ is contained in~$T'$. Thus,~$T'$
  contains each color of~$M'$. Finally,~$T'$ contains each color at
  most once since~$T$ does.
  
  Now assume that~$(M',G',\chi')$ is a yes-instance and let~$S_{T'}$
  be an occurrence of~$M'$ in~$G'$. Let~$T$ denote $G[S_{T'}\cup
  V_I]$, where $V_I$ is the set of vertices removed during Phase~I of
  the kernelization. We show that~$T$ is connected and contains every
  color of~$G$ exactly once. To see that~$T$ is connected observe the
  following: Clearly,~$G[\{r\}\cup V_I]$ is connected. Moreover, each
  vertex~$v\neq r$ of~$T'$ has in~$T'$ a path to~$r$. This path
  contains a subpath from~$v$ to the root~$r'$ of the pendant
  non-unique subtree containing~$v$. In~$G$,~$r'$ is adjacent to some
  vertex of~$\{r\}\cup V_I$. Therefore,~$r'$ is connected to~$r$
  in~$T$ and thus~$T$ is connected. It remains to show that~$T$ 
  contains every color of~$G$ exactly once. Clearly,~$T'$ contains at
  least one vertex of each color~$c\in M'$. Moreover, it also contains
  at least one vertex of each color~$c\in M\setminus M'$ since it
  contains all vertices of~$V_I$. Besides, it contains each color only
  once: The vertices of~$T'$ have pairwise different colors and
  different colors than those of the vertices of~$V_I$. Finally, the vertices
  of~$V_I$ have different pairwise colors since the kernelization did
  not return ``no''.

  The running time can be seen as follows. Determining the pendant
  non-unique subtrees can be done by a standard bottom-up procedure in
  linear time. Removing all vertices during Phase~I can also be
  achieved in linear time. After removing a vertex with color~$c$ in
  Phase~I, we label~$c$ as \emph{occupied}. When we remove a vertex
  with an occupied color during Phase~I, we immediately return
  ``no''. After the removal of vertices during Phase~I, we can
  construct~$M'$ from~$M$ in linear time by removing each occupied
  color. Finally, we can in linear time add an edge between~$r$ and
  every root of a pendant non-unique subtree and then remove all
  remaining vertices that have an occupied color. The final graph~$G'$
  is obtained by performing a depth-first search from~$r$, in order to
  include only those vertices still reachable from~$r$.
\end{proof}

Now, let us turn to developing fast(er) FPT algorithms for~\sCGM.
It can be seen that it is possible to solve~\sCGM in
trees in time~$1.62^\ell\cdot n^{\Oh(1)}$, by 'branching on colors with the most
occurrences' until every color appears at most twice. More precisely,
for a color~$c$ that appears at 
least three times and some vertex~$v$ with color~$c$, we can branch
into the two cases to either delete~$v$ or to delete the at least two
other vertices that have color~$c$. The branching
vector\footnote{For an introduction to the analysis of branching vectors, refer to~\cite{CFK+15,FK10}.} for this
branching rule is~$(1,2)$ or better. Now, if every color appears at most
twice, then \sCGM on trees can be solved in polynomial
time~\cite[Lemma~2]{FFHV11}. 
By a different branching approach, the above running time can be
further improved.  

\begin{brule}
\label{brule:2-subtrees}
If there is a color~$c$ such that there are two vertices~$u$ and~$v$
with color~$c$ that are both not leaves of the tree~$G$, then branch
into the case to delete from~$G$ either
\begin{itemize}
\item the maximal subtree containing~$u$ and all vertices~$w$ such
  that the path from $v$ to $w$ contains~$u$, or
\item the maximal subtree containing~$v$ and all vertices~$w$ such
  that the path from $u$ to $w$ contains~$v$.
\end{itemize}
\end{brule}
\begin{proof}[Proof of correctness.]
  No occurrence may contain vertices of both subtrees, since in this
  case it contains~$u$ and~$v$ which have the same color.
\end{proof}
If the rule does not
apply, then one can solve the problem in linear time; here,
let~$\occ(c)$ denote the number of occurrences of a color~$c$ in~$G$.
\begin{lemma}
  \label{lem:only-leaves}
  Let~$(M,G,\chi)$ be an instance of~\CGM such that~$G$ is a tree and
  for each color~$c$ with~$\occ(c)>1$ at least~$\occ(c)-1$ occurrences
  of~$c$ are leaves of~$G$, then~$(M,G,\chi)$ can be solved in $\Oh(n)$~time.
\end{lemma}
\begin{proof}
  For each color~$c$ with~$\occ(c)>1$, the algorithm simply
  deletes~$\occ(c)-1$ leaves with color~$c$. This can be done in
  linear time by visiting all leaves via depth-first search, checking
  for each leaf in~$\Oh(1)$ time whether~$\occ(c)>1$ and deleting the
  leaf in~$\Oh(1)$ time if this is the case. The resulting graph
  contains each color exactly once, and it is connected since a tree
  cannot be made disconnected by deleting leaves.
\end{proof}
Altogether, we arrive at the following running time.
\begin{theorem}
  \CGM can be solved in~$\Oh(\sqrt{2}^{\halfskip\ell} + n)$ time if~$G$ is a
tree.\label{thm:cgm-tree}
\end{theorem}
\begin{proof}
  The algorithm is as follows. First, reduce the input instance
  in~$\Oh(n)$ time to an equivalent one with~$\Oh(\ell)$ vertices using
  the kernelization of Theorem~\ref{thm:cgm-tree-kernel}. Now, apply
  Branching~Rule~\ref{brule:2-subtrees}. If this rule is no longer
  applicable, then solve the instance in~$\Oh(\ell)$ time (by applying
  the algorithm behind Lemma~\ref{lem:only-leaves}). Since the graph
  has~$\Oh(\ell)$ vertices, applicability of
  Branching~Rule~\ref{brule:2-subtrees} can be tested in~$\Oh(\ell)$
  time. Thus, the overall running time is~$\Oh(\ell)$ times the number
  of search tree nodes. Since each application of
  Branching~Rule~\ref{brule:2-subtrees} creates two branches and
  reduces~$\ell$ by at least two in each branch, the search tree has
  size~$\Oh(2^{\ell/2})=\Oh(\sqrt{2}^{\halfskip\ell})$. The resulting running time
  is~$\Oh(\sqrt{2}^{\halfskip\ell}\cdot \ell+n)$. Furthermore, the factor of~$\ell$ in the
  running time can be removed by interleaving search tree and
  kernelization~\cite{NR00}, that is, by applying the kernelization
  algorithm of Theorem~\ref{thm:cgm-tree-kernel} in each search tree node.
\end{proof}

 \section{Conclusion}
 \label{sec:conclusion}
 In this paper, we have studied the \GM{}, \LGM{} and \CGM{} problems,
 and in particular their behavior in terms of parameterized complexity,
 when the parameter is $\ell=|V|-|M|$, i.e. the number of vertices of $G$ that
 are not kept in a solution.

 We left open the parameterized complexity for parameter~$\ell$ for~\LGM{} on trees, even when
 the vertex-color graph is a forest.
 
 As mentioned in the introduction, parameterization by~$\ell$ may be interesting not only from
 a theoretic, but also from an applied point of view. Unfortunately, for the practically
 relevant case of~\LGM{} we have obtained W[1]-hardness even for very restricted color
 lists~$\L$. Moreover, as noted by Fertin et al.~\cite{FFJ17}, a reduction of Rauf et
 al.~\cite{RRNB13} shows that the variant of~\CGM{} where~$G$ is directed and has edge weights
 is W[1]-hard with respect to~$\ell$. However, the combination of~$\ell$ with further structure
 related to the colors of~$C$ led to tractability results~\cite{FFJ17,FFK18}. It would be
 interesting to identify such color-related structure also for~\LGM{}.
 
\bibliographystyle{plain}
\bibliography{gram-dual}
\end{document}